\documentclass[12pt,a4paper]{amsart}
\usepackage{amsfonts}
\usepackage{amsthm}
\usepackage{amsmath}
\usepackage{amscd}
\usepackage[latin2]{inputenc}
\usepackage{t1enc}
\usepackage[mathscr]{eucal}
\usepackage{indentfirst}
\usepackage{graphicx}
\usepackage{graphics}
\usepackage{pict2e}

\numberwithin{equation}{section}

\theoremstyle{plain}
\newtheorem{Th}{Theorem}

 \theoremstyle{definition}
\newtheorem{Def}[Th]{Definition}

\newtheorem{Rem}[Th]{Remark}
\newtheorem{?}[Th]{Problem}
\newcommand{\R}{\mathbb{R}}

\newcommand{\C}{\mathbb{C}}

\newcommand{\tr}{\operatorname{tr}}

\begin{document}

\title[Classical information storage in a quantum 
system]
{Classical information storage in an $n$-level quantum 
system}

\author[P. E. Frenkel]{P\'eter E. Frenkel}

\address{E\"{o}tv\"{o}s Lor\'{a}nd University \\ Department of Algebra and Number Theory
 \\ H-1117 Budapest
 \\ P\'{a}zm\'{a}ny P\'{e}ter s\'{e}t\'{a}ny 1/C \\ Hungary}
\email{frenkelp@cs.elte.hu}

\author[M. Weiner]{Mih\'aly Weiner}
\address{Budapest University of Technology and Economics \\ Department of Analysis
 \\ H-1111 Budapest
 \\  M\H{u}egyetem rkp. 3--9\\ Hungary}
\email{mweiner@renyi.hu}

\thanks{The first author's research is partially supported by MTA R\'enyi
``Lend\"ulet'' Groups and Graphs Research Group and OTKA grant no.\ K109684. The second author's research is supported in part by the ERC Advanced Grant 227458 ``Operator Algebras and
Conformal Field Theory'', OTKA grant no.\ 104206 and by the ``Bolyai J\'anos'' Research Scholarship of the Hungarian Academy of Sciences.} 

 \dedicatory{Dedicated to Professor Andor Frenkel on the occasion of his $80^{\text{th}}$ birthday}

 \subjclass[2010]{}


\begin{abstract}A game is played by a team of two 
--- say Alice and Bob --- in which the value of a random variable $x$
is revealed to Alice only, who cannot freely communicate with Bob.
Instead, she is given a quantum $n$-level system, respectively a 
classical $n$-state system,  which she can put in possession of Bob in any
state she wishes. We evaluate how successfully they managed to 
store and recover the value of $x$ 
 by requiring 
Bob to specify a value $z$ and giving a reward of value $  f(x,z)$ to the team. 

We show that whatever the probability distribution of $x$ and the 
reward function $f$ are, when using a quantum $n$-level system, 
the maximum expected reward obtainable with the best possible 
team strategy is 
equal to that obtainable with the use of a classical $n$-state system.
The proof relies on mixed discriminants of positive matrices and ---
perhaps surprisingly --- an application of the Supply--Demand Theorem for bipartite graphs. 

As a corollary,  we get an infinite set of new, 
dimension dependent inequalities regarding positive operator valued measures and density operators on complex $n$-space.

As a further corollary, we see that the greatest value, with respect to a given distribution of $x$, of the mutual information $I(x;z)$ that is obtainable using an $n$-level quantum system
  equals the analogous maximum for a  classical $n$-state system. 
\end{abstract}

\maketitle

\section{Introduction}

In contrast to a classical bit which has only $2$ pure states, a
qubit has infinitely many. However, this does not necessarily mean that 
we can store more (classical) information in a qubit than in a classical bit.
The point is that although the qubit has infinitely many different pure 
states, it is impossible to distinguish these states with {certainty}.
This is a fundamental fact, and cannot be circumvented by some better
measuring device.

In the case of a qubit, one can distinguish with certainty between at most
$2$ states. In general, in the case of an $n$-level system (whose state 
space is modelled by the set of density operators of an $n$-dimensional complex
Hilbert space),  one can distinguish with certainty between at most $n$ states.
So in this respect a quantum $n$-level system performs 
 like a classical $n$-state system. However, this does not make them 
necessarily equivalent. Perhaps it is possible to distinguish between 
$k>n$ states of a quantum $n$-level system not with certainty, 
but in a way that is --- in some sense --- ``closer'' to certainty than 
what we can achieve with a classical $n$-state system. 

How to define ``closer to certainty''? In general, we 
may view our qubit as a memory --- or as is often done in the literature: 
as a channel\footnote{
Usually by specifying a {\it channel} one means to fix a collection 
of states corresponding to encoding, while the measurement on the decoding side
remains unspecified. For us neither encoding nor decoding is fixed as both 
are  to be optimalized; only the level $n$ of the system (i.e.\! the 
dimension of the Hilbert space) is fixed. For this reason we prefer to talk
about a {\it memory unit} rather than a channel.}
 --- in which there is an ingoing and an outgoing information:
Alice chooses a certain state from a previously fixed set of states and 
puts the system into the selected state, then passes it to Bob, who 
will have to try to figure out the chosen state. 
So one may investigate the 
issue from the point of view of some kind of (classical)
{\it channel capacity}. 
One idea is of course to use  Shannon-type 
informational capacity which is well-investigated in the quantum setting; 
see e.g.\! \cite{quantumbook} on quantum information theory and \cite{OP} on quantum entropy; see also Section~\ref{mutual} of the present paper.

However, here we argue that this in itself cannot fully settle the 
problem. As an example, suppose that the following game is played.

A \$1 bill is put randomly and with equal probability into one of  
$3$ boxes. Bob will  pick one of the boxes and  he will  get what is inside that 
box. Alice knows where the \$1 bill has been put, and she wishes to help Bob. 
However, Alice is not allowed to  directly tell  Bob where the money is (in which case Bob could 
always get the \$1 bill with certainty). Instead, she is only allowed to 
send to Bob a classical bit, respectively  a qubit (not previously entangled to anything 
else) whose state she can manipulate as she wishes. That is, she is allowed 
to send a classical or a quantum bit 
of information.

They may agree on some scheme beforehand. For example, played with a 
classical bit, Alice and Bob can agree that the bit-value $0$ will mean 
that the money is in box nr.\! $1$ (in which case Bob will pick
box nr.\! $1$) and the bit-value $1$ will mean that 
the money is   either in box nr.\! $2$ or in box nr.\! $3$ 
(in which case Bob will toss a coin and accordingly pick 
box nr.\! $2$ or $3$). The question then becomes: after the game is
played once, what is the expected value of the money won? 

Every team-strategy leads to a specific {\it channel matrix}, i.e.\,  
a collection of conditional probabilities of the type $a_{ij}$ 
where 
$$
a_{ij}=P(\textrm{Bob chooses the } i^{\rm th}\; 
\textrm{box given that the money is in the } j^{\rm th}\; \textrm{box}),
$$
so that for example the previously described simple strategy using a classical
bit gives the channel matrix
$$
\left(
\begin{matrix}
1 & 0 & 0 \\
0 & 1/2 & 1/2 \\
0 & 1/2 & 1/2
\end{matrix}\right).
$$
In general, the channel matrix $A$ will be a {\it stochastic matrix}:
its entries are all nonnegative, and each column sums to $1$. The
above channel matrix will allow Bob to have an expected win of $2/3$ 
dollars, so we may say that its ``money capacity'' is $c_\$=2/3$. It is 
an elementary exercise to show that this is the best we can get using
a classical bit.

Note that in general $c_\$(A) = \frac{1}{3}{\rm tr}(A)$ whereas 
the ``usual'' (informational) capacity $c(A)$ would be the maximum 
{\it mutual information}
between Bob's choice and the actual place of the money; for the above 
channel matrix $c$ is precisely $1$ bit. Now, by Holevo's celebrated 
theorem \cite{holevopaper}, even if Alice and Bob used a quantum bit instead 
of a classical one, the informational capacity $c$ could not get above 
$1$ bit. This is in accordance with the (common) belief that a single
qubit (on its own) is  worth no more than a classical one.

The fact that in  {\it superdense coding} \cite{densecoding}
Alice manages to transmit $2$ bits of information to Bob by 
physically sending only $1$ qubit is no contradiction to what 
was said. 
Indeed, in superdense coding $2$ qubits are used:
for the {\it decoding} part both of them are necessary. However, 
for the {\it encoding} part only one of them is needed; this is 
achieved by previously entangling the $2$ qubits. So the 
$2$ classical bits are actually stored in $2$ quantum bits; the 
surprising feature is rather that it is a kind of 2-bit 
memory made out of 2 qubits where for the ``read out'' we need both, 
but for the ``write in'' we only need to get in touch with one of them.

However, here we are interested by how much (if any) better is a qubit 
{\it on its own} (not entangled to other qubits) than a classical one.
So can we apply Holevo's theorem to conclude that in the described 
game, even by using a qubit, Alice and Bob cannot win more than 
$2/3$ of a dollar (i.e.\ the maximum amount possible when a classical bit
is used)? The answer is negative. In fact, consider the stochastic matrix
$$
A=\left(
\begin{matrix}
3/4 & 1/8 & 1/8 \\
1/8 & 3/4 & 1/8 \\
1/8 & 1/8 & 3/4
\end{matrix}\right).
$$
Its ``money capacity'' $c_\$(A) = \frac{1}{3}\tr A=\frac{3}{4}$ 
is larger than what can be achieved by using a classical bit. However, 
elementary arguments together with a straightforward computation show that 
$c(A) = \frac{7}{4}{\rm log}(3)-\frac{9}{4}{\rm log}(2)$, which is 
{\it smaller} than ${\rm log}(2)$, i.e., smaller than one bit.
This should not be considered a surprise: in Shannon's
{\it Noisy Channel Coding Theorem}, the channel capacity as 
reliable transmission rate is only achieved in the ``long run''; 
with a single use of the channel, things can go different.
 
Thus, as shown by the previous example, Holevo's theorem cannot rule 
out the existence of a strategy by which using a qubit Alice and 
Bob can win more money in this game (in terms of expected values) than 
what is possible using a classical bit. 

Nevertheless, for this actual game it is not difficult to come up with 
an argument \cite{Aidan} to show that $c_\$ \leq 2/3$ holds even if Alice and 
Bob use a qubit. However, why sticking exactly to this game, and why 
investigating only the case of $2$-level systems (i.e.\! single bits)?

In general we might consider the following scheme. The value of a 
random variable $x$ is revealed to Alice but not to Bob. 
Though previous to the game they can meet and agree to follow any 
kind of strategy they like, during the game Alice cannot
freely communicate with Bob. Instead, she is given a quantum $n$-level system 
(or alternatively: a
classical $n$-state system) which she can put in possession of Bob in any
state she wishes. Bob is then required to specify a value $z$.
We evaluate how successfully they managed to
store and recover the value of $x$ by giving a reward of value $ f(x,z)$ 
to the team, where $f$ is some previously fixed ``reward function'' 
(e.g.\ in our original game the reward was \$1 if $x=z$ and zero 
otherwise).

Consider the sets $\mathcal C_n(k,l)$ and $\mathcal Q_n(k,l)$ defined as convex hulls
of all possible $k \times l$ channel matrices that can be obtained by 
Alice passing to Bob a classical $n$-state system or a quantum $n$-level 
system, respectively. We postpone a more detailed description of these
sets to the next section, but note that obviously $\mathcal C_n(k,l) \subseteq 
\mathcal Q_n(k,l)$.

Since we think of the distribution of $x$ as 
fixed with the rules of the game, 
the expected reward ${\mathbb E}(f(x,z))$ is just an arbitrary affine
 linear functional of the channel matrix of Alice and Bob 
(e.g.\ in our original game it was $1/3$ times the trace). Thus, there would exist a game of the specified type in which it is more efficient 
to use a quantum $n$-level system than a classical one if and only if we had 
$\mathcal C_n(k,l) \ne \mathcal Q_n(k,l)$. However, our main result is that 
$\mathcal C_n(k,l) = \mathcal Q_n(k,l)$ for all values of $n$, $k$, and $l$.

Note that this result
is  trivial  for $n\ge\min (k,l)$, since then both sets consist of all stochastic $k\times l$ matrices. However, in general the equality  is nontrivial in the sense that it results in some new, dimension
dependent inequalities regarding positive operator valued measures and density matrices.

It is worthwile to make a remark on the specific type of game we are dealing with.
Suppose for a change  that $x$ and $y$ are drawn uniformly  and independently from the set $\{A,B,C\}$ and $x$ is revealed to Alice (but not to Bob) whereas the value of $y$ is revealed to Bob (but not to Alice). This time they get rewarded if Bob correctly guesses whether $x=y$ or not, but again, they can only communicate in a certain restricted way: Alice is allowed to pass to Bob a single quantum bit (or alternatively: a single classical bit). Now it is easy to show that in this game, using a quantum bit rather than a classical one is indeed more advantageous. This does not contradict our result --- this game is not of the discussed type: here Bob recieves information from {\it two} separate sources (apart from what he gets from Alice, he also receives the value of $y$). So again we stress that our result concerns the capacity of a quantum $n$-level system when it is considered {\it on its own}, without further supporting classical or quantum information (like the value of $y$ in the new game). The new  game is in fact a {\it quantum fingerprinting} problem. For such problems,  a quantum system is known \cite{quantumfingerprinting} to perform exponentially better than the corresponding classical system.

\bigskip\noindent
\bf Notations and terminology. 
\rm The set $\{1,\dots, k\}$ is denoted by  $[k]$. We write $e_{ij}$ for the matrix that has an entry 1 at position $(i,j)$ and all other entries zero. The identity matrix is $\bf 1$. A matrix is \it stochastic
\rm if all entries are nonnegative reals and each column sums to 1.
A complex  matrix $A$  is \it psdh \rm if it is positive semidefinite Hermitian, written $A\ge 0$. A \it positive operator valued measure (POVM), \rm  also called a \it partition of unity, \rm is a sequence $E_1$, \dots, $E_k$ of psdh matrices summing to $\bf 1$.
A \it density matrix \rm is a psdh matrix with trace 1.

\section{Sets of channel matrices}
Throughout this section, let $n$, $k$ and $l$ be fixed positive integers.

Suppose that a letter of an alphabet containing $l$ letters is to be encoded 
by Alice in a classical $n$-state system, whereas on the decoding side
Bob uses an alphabet containing $k$ letters. Which  channel matrices
can Alice and Bob realize by a suitable  strategy? 

Every strategy is a convex combination of {\it pure} strategies; strategies
--- we are in the classical setting --- in which no randomness appears. Thus 
in the  case of a pure strategy, the channel matrix is a $k\times l$ matrix such
that
\begin{enumerate}
\item
each entry is either $1$ or $0$,
\item
in each column there is exactly one $1$, 
\item
the number of nonzero rows is at most $n$. 
\end{enumerate}
This last property is due to the fact that Alice and Bob use a 
classical $n$-state system: if no randomness is used, then at
decoding at most $n$ different things can happen, regardless of the number $k$ of letters that  Bob has in his alphabet.
Thus we make the following

\begin{Def}
Let $\mathcal C=\mathcal C_n(k,l)$ be the convex hull of $k\times l$  matrices satisfying properties (1), (2) and (3) above. 
\end{Def}

Then $\mathcal C$ is a convex polytope whose vertices are the $k\times l$  matrices satisfying (1), (2) and (3). 
Note that $\mathcal C$ is also the convex hull of $k\times l$ stochastic  matrices with at most $n$ nonzero rows.

Now suppose that instead of a classical $n$-state system, Alice can
use an $n$-level quantum system. Its state space can be identified 
with the set of complex $n\times n$ {\it density matrices}: the set
of matrices $\rho\in M_n(\C)$ such that
$$
\rho \geq 0, \; {\tr}(\rho)=1.
$$ 
A specific measurement scheme with $k$ possible 
outcomes gives rise to an affine map from this state space to the set of 
classical probability distributions on the set $[k]=\{1,2,\ldots, k\}$. Such 
an affine map is always given in the following way: we have a
{\it positive operator valued measure} (POVM) 
$E_1,E_2,\ldots, E_k \in M_n(\C)$, i.e.\! $E_i\geq 0$ for each 
$i=1,2,\ldots , k$ and $E_1+E_2+\ldots +E_k = \bf 1 $ (identity matrix), and the 
map in question is
\[
\rho \mapsto \left({\tr}(E_1\rho ), {\tr}(E_2\rho ), \ldots 
{\tr}(E_k\rho)\right),
\]
see details in \cite[Section 1.6]{holevobook}. Thus the most general
measurement with $k$ outcomes is a POVM $E_1$, $E_2$, \dots, $E_k$ in the sense
that if the state of the system was described by the density operator
$\rho$ then the measurement will result in the $i^{\rm th}$ outcome 
with probability ${\tr}(E_i\rho )$.

Now let us return to the set of possible channel matrices. On the 
encoding side, Alice needs to choose a state for each letter to be 
encoded. On the decoding side, Bob needs to choose a measurement whose
result will be interpreted as ``read out''. Thus a specific strategy is
given by the choice of $l$ density matrices $\rho_1, \rho_2 , \dots, \rho_l 
\in M_n(\C)$ and a POVM $E_1,E_2,\ldots, E_k\in M_n(\C)$ resulting in 
the channel matrix $A$ with entries
$$
a_{ij} = \tr (E_i\rho_j).
$$
The set of matrices $A$ we can obtain depends on $n$, $k$ and $l$.
We make the following

\begin{Def}Let $\mathcal Q=\mathcal Q_n(k,l)$ be the convex hull of $k\times l$  matrices of the form $\left(\tr (E_i\rho_j)\right)$, where
 $E_1, \dots, E_k\in M_n(\C)$ is  a POVM and $\rho_1, \dots, \rho_l\in M_n(\C)$ are density matrices.
\end{Def}
It may not be obvious that $\mathcal Q$ is a polytope, but in fact, our main result is
\begin{Th}\label{main}
$\mathcal C=\mathcal Q$.
\end{Th}
\noindent  We start by proving the trivial inclusion.

\bigskip
\noindent
\bf Proof of $\mathcal C\subseteq\mathcal Q$. \rm
Assuming $A\in \mathcal C$, we show that $A\in \mathcal Q$.
We know  that $A$ is a stochastic matrix, and we may assume that
 only  the first $\min (n,k)$ rows of $A$ can be nonzero.

Put \[\rho_j=\sum_{i=1}^{\min(n,k)}a_{ij}e_{ii}.\] These are density matrices.

When $n\le k$, put $E_i=e_{ii}$ for $i\le n$ and $E_i=0$ otherwise. 
When  $n\ge k$, put $E_i=e_{ii}$ for $i\le k-1$ and $E_k=\sum_{i=k}^n e_{ii}$. 
In either case, this is a POVM. We have $\tr (E_i\rho_j)=a_{ij}$, whence $A\in \mathcal Q$. 
\qed

\bigskip
\noindent
For the proof of the reverse inclusion, we  recall the definition and the positivity property of mixed discriminants.

The determinant is a homogeneous polynomial of degree $n$ on $M_n(\C)$. 
Therefore,  there exists a unique symmetric $n$-linear
function $D$ such that
$$
D(X,\ldots , X) = \det X
$$
for all $X\in M_n(\C)$. This function $D$ is the  
{\it mixed discriminant}. By \cite[Lemma 2(vi)]{B}, if
$E_1$, \dots, $ E_n$ are all positive semidefinite Hermitian  matrices, then
$$
D(E_1,\ldots, E_n)\geq 0.
$$

\bigskip\noindent\bf Proof of $\mathcal Q\subseteq\mathcal C$. \rm
Assume that $A\in \mathcal Q$. We prove that $A\in \mathcal C$.
We may assume that
 $a_{ij}= \tr (E_i\rho_j)$, where  $E_1, \dots, E_k\in M_n(\C)$ is  a POVM and $\rho_1, \dots, \rho_l\in M_n(\C)$ are density matrices.

For $I=(i_1,\dots, i_n)\in [k]^n$, put
\[p_I=D(E_{i_1},\dots, E_{i_n}),\] where $D$ is the mixed discriminant
.
We have $p_I\ge 0$ for all $I$. Thus, we get a  measure $P$
on $[k]^n$ defined by $P(S)=\sum_{I\in S}p_I$.
Using the multilinearity of $D$ and the assumption that $E_1$, \dots, $E_k$ is a partition of unity, we see that
\[P([k]^n)=D(\bf 1,\dots, \bf 1)=\det\bf  1=\rm 1,\]
so $P$ is  a probability measure.
In fact, for any $R\subseteq [k]$, we may put $E_R=\sum_{i\in R} E_i$, and then  we have
\[P(R^n)=\det E_R.\] Since $0\le E_R\le\bf 1$, all eigenvalues of $E_R$ are in $[0,1]$.
Thus, $\det E_R$, the product of eigenvalues, does not exceed the smallest eigenvalue. Hence,
$(\det E_R){\bf 1}\le E_R$, so, for all $j$,
\[\tr (E_R\rho_j)\ge \tr ((\det E_R){\bf 1} \rho_j)=\det E_R.\]
The left hand side here is $A_j(R)$, where $A_j$ is the probability measure on $[k]$ given by the $j$-th column of $A$. So we have
\[A_j(R)\ge P(R^n)\qquad\textrm{ for all }R\subseteq [k].\]

Let us connect $I\in [k]^n$ to $i\in [k]$ by an edge if $i$ occurs in $I$. This gives us a bipartite graph. The neighborhood of any set $S\subseteq [k]^n$ is the set $R\subseteq [k]$ of indices occurring in some element of $S$. We always have $S\subseteq R^n$, whence
$A_j(R)\ge P(R^n)\ge P(S)$. Thus, by the Supply--Demand Theorem, there exists a probability measure $\tilde P_j$ on $[k]^n\times [k]$ which is supported on the edges of the graph and has marginals $P$ and $A_j$. Whenever $p_I\ne 0$, let $B(I)$ be the $k\times l$ stochastic matrix whose $j$-th column is given by the conditional distribution  $\tilde  P_j|I$ on $[k]$. Then $B(I)$ has at most $n$ nonzero rows, and $A=\int BdP\in \mathcal C$.
\qed

\begin{Rem} (Infinite channel matrices.)  We may replace  $[l]$ by any subset $J$ of an affine space. 
 In this setting, the polyhedron of $k\times l$ stochastic matrices is replaced by the set  of those  nonnegative functions on $[k]\times J$ that sum to 1 on $[k]$ for any fixed $j\in J$ and are piecewise linear on $J$ for any fixed $i\in [k]$.
 Let $\mathcal C_n(k,J)$ be the  convex hull of elements that are supported on $N\times J$ for some $N\subseteq [k]$ of cardinality at most $n$. 
 Let $\mathcal Q_n(k,J)$ be the convex hull of all elements of the form $(i,j)\mapsto\tr (E_i\rho(j))$, where $E_1$, \dots, $E_k$ is a POVM  and $\rho$ is a piecewise linear map from $J$ to the set of $n\times n$ density matrices.  Then $\mathcal C_n(k,J)=\mathcal Q_n(k,J)$. This follows easily from the proof of Theorem~\ref{main}, so this is left to the reader.

Going one step further, we may also replace  $[k]$ by a separable metric 
 space $X$. Let $\mathcal M(X)$ be the set of 
 probability measures on the Borel sets of $X$, endowed with the Prokhorov metric. In this setting, the polyhedron  of $k\times l$ stochastic matrices is replaced by the set of piecewise linear functions $J\to\mathcal M(X)$, endowed with the supremum metric. 
 Let $\mathcal C_n(X,J)$ be the closure of the convex hull of elements whose range is contained in   $\mathcal M(N)$ for some  $N\subseteq X$ of cardinality at most $n$.   Let $\mathcal Q_n(X,J)$ be the closure of the convex hull of all elements of the form $j\mapsto\tr (E\rho(j))$, where $E:\mathcal B(X)\to M_n(\C)$ is a positive operator valued 
 probability measure on the Borel sets of $X$ and $\rho$ is a piecewise linear map from $J$ to the set of $n\times n$ density matrices. 
 Then $\mathcal C_n(X,J)=\mathcal Q_n(X,J)$. This follows easily from the previous paragraph and the fact that finitely supported POVMs are dense (due to the separability of  $X$). 

Everything said above remains true if $J$ is any set, resp.\ a topological space, and the words `piecewise linear' are erased, resp.\ replaced by `continuous' on all occurrences.  
\end{Rem}

\section{Inequalities for POVM's and density matrices}
As before, let  $n$, $k$ and $l$ be positive integers. If a linear inequality is satisfied by the entries  of any $k\times l$  stochastic 0-1 matrix with at most $n$ nonzero rows, then we can use Theorem~\ref{main} to deduce that it is also satisfied by the entries of any $A\in \mathcal Q_n(k,l)$.  
This is a way to get new inequalities for POVM's and density matrices.
Therefore, we want to find  inequalities satisfied by all $A\in \mathcal C_n(k,l)$.  

When $n\ge\min (k,l)$, the polytope $\mathcal C_n(k,l)$ is obviously the set of all stochastic $k\times l$ matrices and we do not get anything interesting.

In general, we are not
able to describe the faces of the polytope $\mathcal C_n(k,l)$. However, it is clear that a $k\times l$ real  matrix 
$A$ belongs to the polytope if and only if
it satisfies all linear inequalities \begin{equation}\label{ineq}\tr (CA)\ge c \qquad  (C\in \R^{l\times k}, c\in\R)\end{equation} that the vertices satisfy. The vertices are the stochastic 0-1 matrices $A$ with at most $n$
nonzero rows,  and  all of them  satisfy \eqref{ineq} if and only if { for all }$ N\subseteq [k]$, $|N|=n$, we have
\begin{equation}\label{testing} \sum_r\min_{i\in N} c_{ri}\ge c\end{equation} for the entries of the matrix $C=(c_{ri})$. For example, if $C$ is a 0-1 matrix such that any $n$ columns have at least one 1 at the same position, and $c=1$, then 
 the inequalities \eqref{testing} hold, and therefore $\tr (CA)\ge 1$ for all $A$ in the polytope. E.g., let $n=2$ and \begin{equation}\label{CA}C=\begin{pmatrix} 0&1&1&1\\1&1&0&0\\1&0&1&0\\1&0&0&1\end{pmatrix},\qquad A=\begin{pmatrix} 1/2&0&0&0\\1/6&0&1/2&1/2\\1/6&1/2&0&1/2\\1/6&1/2&1/2&0\end{pmatrix},\end{equation} then $\tr (CA)=1/2$, so $A$ is not in the polytope $\mathcal C_2(4,4)=\mathcal Q_2(4,4)$.

Now observe that if $C$ is  a matrix of size $m\times k$, with arbitrary $m$, such that the inequalities \eqref{testing} hold, then any vertex, and therefore any point $A$ of the polytope $
\mathcal C$
satisfies  
\begin{equation}\label{strong}\sum_r\min_{j\in [l]} (CA)_{rj}\ge c,\end{equation} which is stronger than \eqref{ineq}.  Note that \eqref{strong} can be rewritten as a  system of $l^r$ linear inequalities holding simultaneously. Namely, we choose a $j_r$ for each  $r$ and write
\begin{equation}\sum_r (CA)_{rj_r}\ge c.\end{equation}

As an example, choose $n\le r\le k$,  let $m=\binom kr$, and let the rows of $C$ be indexed by the $r$-element subsets $S$ of $[k]$. Let $c_{S,i}=1$ if $i\in S$ and zero otherwise. Then the inequalities \eqref{testing} hold with \[c=\binom{k-n}{r-n}=\binom{k-n}{k-r}\] and therefore \eqref{strong} also holds, i.e.,\[\sum_{|S|=r}\min_{j\in [l]}\sum_{i\in S}a_{ij}\ge
\binom{k-n}{k-r}\]  for all $A$ in the polytope.  Replacing $r$ by $k-r$
and $S$ by $[k]-S$, and using that $A$ is stochastic, we get
\begin{equation}\label{r}\sum_{|S|=r}\max_{j\in [l]}\sum_{i\in S}a_{ij}\le\binom kr-
\binom{k-n}{r}\end{equation}  for all $A$ in the polytope, whenever $0\le r\le k-n$. When $r=1$, this is immediate for $A\in\mathcal Q_n(k,l)$ from the fact that for any density matrix $\rho$, we have $\rho\le \bf 1$ and so $\tr (E\rho)\le\tr E$ for any psdh matrix $E$.  However, when $r\ge 2$, the inequalities \eqref{r} seem nontrivial for $A\in\mathcal Q_n(k,l)$, and they only follow from Theorem~\ref{main}. Note that the  inequalities \eqref{r} are not sufficient to describe the polytope. Indeed, for $k=l=4$ and $n=2$, the matrix $A$ in \eqref{CA} satisfies \eqref{r} for all $r$, but is not in the polytope.

These examples already show that we have a huge freedom in choosing $C$, and the combinatorics of  families of subsets of $[k]$ enters into  the subject of finding dimension dependent linear  inequalities for values $\tr (E_i\rho_j)$. This could  turn out to be interesting.
For example, think of the famous question of {\it mutually unbiased bases}.
A complete set of mutually unbiased bases can also be described as a set
of density operators $\rho_1,\ldots, \rho_{n(n+1)}\in M_n(\C)$ with
$({\tr(\rho_i\rho_j)})_{i,j=1}^{n(n+1)}$ being certain prescribed 
values; see e.g.\! \cite{polytope} for a good review or \cite{en} for some 
recent development. If $n$ is a power of a prime, then such systems exist, 
while for other dimensions such systems are believed to not to exist, although this  
has not yet been proved. Thus, to prove nonexistence, one will have to use 
inequalities (or other tools) that show nontrivial dependence on the dimension $n$.

To close this section, we mention the slightly related open question of describing the cone of \it completely positive semidefinite \rm matrices. These are   $k\times k$  matrices of the form  $X=(\tr (A_iA_j))$, where $k$ is fixed, $n$ is arbitrary, and $A_1$, \dots, $A_k$ are positive semidefinite $n\times n$ matrices. (Note that requiring the $A_i$ to be \it real \rm matrices rather than \it  complex \rm ones leads to the same notion of complete positive semidefiniteness.)  Clearly, $X$ is  positive semidefinite, with
nonnegative real entries. However, not all positive semidefinite and entrywise nonnegative matrices are completely positive semidefinite, even though $n$ was arbitrary in the definition; see \cite{F, mi,  LP}.

\section{Mutual information}\label{mutual}
We now wish to place Theorem~\ref{main} in context with respect to Holevo's inequality.
First, we recall some basic information theory.
The \it Shannon entropy \rm of a sequence $p=(p_1,\dots, p_n)$ of nonnegative reals summing to 1 (i.e., a probability distribution) is defined to be \[H(p)=\sum_{i=1}^np_i\log\frac 1{p_i}.\]  We have \begin{equation}\label{log}0\le H(p)\le\log n.\end{equation}
The \it  von Neumann entropy \rm $S(\rho)$ of a  density matrix $\rho$ is the Shannon entropy of its eigenvalues. Von Neumann entropy is a  concave function on density matrices.  For density matrices $\rho_1$, \dots, $\rho_l$ of size $n\times n$ and  a probability distribution $q=(q_1, \dots, q_l)$, define \[\chi=S\left(\sum q_j\rho_j\right)-\sum q_jS(\rho_j).\]  We have \begin{equation}\label{chirange}0\le\chi\le \min(H(q),\log n).\end{equation}

When $x$ is a discrete random variable with distribution $p$, we write $H(x)=H(p)$.  Given a    matrix $M$ whose elements are nonnegative and sum to 1, we may view $M$ as the distribution of a pair $(z,x)$. The mutual information between $z$ and $x$ is defined to be \[I=H(z)+H(x)-H(z,x),\] which is symmetric with respect to $z$ and $x$. We have \begin{equation}\label{infobound} 0\le I\le \min (H(z), H(x)).\end{equation}

In quantum information theory, we are interested in the case \[M=(q_j\tr(E_i\rho_j)),\] where     $E_1, \dots, E_k\in M_n(\C)$ is  a POVM, $\rho_1, \dots, \rho_l\in M_n(\C)$ are density matrices, and  $q=(q_1, \dots, q_l)$ is a probability distribution. (Cf.\ the game discussed in the Introduction.  The distribution $q$ is the distribution of the random variable $x$ whose value is revealed to Alice.)  We then have

\bigskip
\noindent\bf Holevo's inequality~\cite{holevopaper}. \rm $I\le\chi.$

\bigskip
From Theorem~\ref{main}, we can deduce a different upper bound for $I$.
Let $m=\min (n,k)$ and \begin{equation}\label{psidef}\psi=\max H(Q_1,\dots Q_m),\end{equation} where the maximum is taken over all partitions of $[l]$ into $m$ pairwise disjoint subsets $L_1$, \dots, $L_m$ (empty set allowed), and \[Q_r=\sum_{j\in L_r} q_j \qquad (r=1,\dots, m).\]  We have \begin{equation}\label{psirange}0\le\psi\le \min(H(q),\log m).\end{equation}

\begin{Th}\label{psi} $I\le\psi$.
\end{Th}

\begin{proof}
It is well known that the mutual information for the  matrix $M=(q_ja_{ij})$ is a convex function of the stochastic matrix $A=(a_{ij})$ if the probability distribution    $q=(q_1, \dots,q_l)$ is fixed.  Thus, for $a_{ij}=\tr (E_i\rho_j)$, in which case we have $A\in\mathcal Q_n(k,l)=\mathcal C_n(k,l)$, the maximum of $I$ is attained when $A$ is a  vertex of this polytope, i.e., a  stochastic 0-1 matrix with at most $m$ nonzero rows. For such $A$, we have
 \[I=H(z)+H(x)-H(z,x)=H(z)\le\psi,\]
since $x$ and $(z,x)$ both have distribution $q$ and $z$ has a distribution of the form $Q_1$, \dots, $Q_m$.
\end{proof}

\begin{Rem}\label{notimply} In view of \eqref{chirange} and \eqref{psirange}, either one of Holevo's inequality and Theorem~\ref{psi} implies $I\le 
\log n
$.

 It is obvious that Theorem~\ref{psi} does not imply Holevo's inequality, i.e., it is possible to have $\chi< \psi$. 

We now show that Holevo's inequality does not imply Theorem~\ref{psi}, even if we assume that $n\le k$, so that $m=n$.  For example, let $n=2$, $l=3$, and let $\rho_j$ $(j=1,2,3)$ be projections onto three lines in the plane $\R^2$ that pass through the origin and mutually form angles of $\pi/3$. Let $q_j=1/3$ $(j=1,2,3)$. Then $S(\rho_j)=0$, so \[\chi=S\left(\sum_{j=1}^3 q_j\rho_j\right)=S({\bf 1}/2)=\log 2,\] while \[\psi=H\left(\frac13,\frac23\right)=\frac13\log 3 +\frac23\log\frac32=\log\frac3{\sqrt[3]4},\] whence $\psi<\chi$ for this example.  
\end{Rem}

We now have two upper bounds for the mutual information $I$, of very different nature: Holevo's inequality is analytic, Theorem~\ref{psi} is combinatorial.  A common lower bound for $\chi$ and $\psi$ is \[\omega=\max \left( S\left(\sum_{r=1}^m Q_r\bar\rho_r\right)-\sum_{r=1}^m Q_rS(\bar\rho_r)\right),\] where the maximum and the $Q_r$ are to be understood as in the definition \eqref{psidef} of $\psi$, and \[\bar\rho_r=\frac1{Q_r}\sum_{j\in L_r}q_j\rho_j\qquad (1\le r\le m,\quad Q_r\ne 0).\]
Indeed, for all partitions $L_1$, \dots, $L_m$, we have 
 \[ S\left(\sum_{r=1}^m Q_r\bar\rho_r\right)-\sum_{r=1}^m Q_rS(\bar\rho_r)\le H(Q_1,\dots, Q_m)\] by \eqref{chirange},
whence $\omega\le\psi$. Also,  $\omega\le\chi$ by concavity of $S$. 

In an earlier version of this paper, we conjectured that
$I\le\omega$ always holds. We are indebted to  Michele Dall'Arno for pointing out to us that this fails already for the projections $\rho_j$ and probabilities $q_j=1/3$  $(j=1,2,3)$ discussed in Remark~\ref{notimply} above, if we choose $E_i=(2/3)({\bf 1}-\rho_i)$ $(i=1,2,3)$, cf.\ \cite{sasaki}.  It is unclear whether there is an upper bound for $I$ that is given by a  `natural' formula and implies both Holevo's inequality and Theorem~\ref{psi}. 
\bigskip

\noindent
{\bf Acknowledgements.}
The authors wish to thank Aidan J.\ Klob\-uch\-ar (a student of the second author) for writing up the game with the \$1 bill~\cite{Aidan}, and M\'arton 
Nasz\'odi and G\'abor Tardos for useful comments and conversations.

\end{document}